\documentclass{iopjournal}

\usepackage{amsmath,amssymb,amsfonts,amsthm,bbm}
\usepackage{bm}
\usepackage{braket}
\usepackage{tikz}
\usepackage{yquant}

\newtheorem{theorem}{Theorem}
\newtheorem{lemma}[theorem]{Lemma}
\newtheorem{proposition}[theorem]{Proposition}

\theoremstyle{definition}
\newtheorem{definition}[theorem]{Definition}

\DeclareMathOperator{\tr}{Tr}

\DeclareMathOperator{\real}{Re}

\DeclareMathOperator{\qft}{QFT}
\newcommand{\id}{\mathbbm{1}}

\begin{document}

\articletype{Paper}

\title{Approximate error correction for quantum simulations of SU(2) lattice gauge theories}

\author{Zachary P Bradshaw}

\affil{QodeX Quantum, Inc., Chicago, IL 60601, USA}

\email{zak@qodexquantum.ai}


\begin{abstract}
We present a protocol for actively suppressing Gauss law violations in quantum simulations of SU(2) lattice gauge theory. Mid-circuit measurements extract a syndrome $(J,M,N)$ characterising the gauge-violation sector at each vertex by resolving both the total angular momentum and the magnetic quantum numbers of the violation through a group quantum Fourier transform. A syndrome-conditional recovery operation maps the state back to the gauge-invariant subspace, and the procedure is iterated as a sweep over vertices in a process we call gauge cooling. We prove that every single-qubit Pauli error at a coordination-four vertex with four spin-$1/2$ edges is detected by the gauge syndrome, and we show that the Knill--Laflamme conditions fail for syndrome-based recovery alone whenever the singlet multiplicity exceeds one. The residual physical-subspace errors carry a structured Pauli decomposition with vanishing $Y$ component, which suggests compatibility with concatenation by a CSS stabilizer code. We demonstrate the protocol on a single-plaquette simulation of the Kogut--Susskind Hamiltonian truncated to the spin-$1/2$ representation under depolarising and amplitude damping noise, and we observe that gauge cooling restores approximate gauge invariance and improves fidelity at noise rates representative of current superconducting hardware.
\end{abstract}

\section{Introduction}
\label{sec:intro}

The standard model of particle physics~\cite{gaillard1999standard} relies on non-abelian gauge symmetry as an organising principle~\cite{tong2018gauge}. Quantum chromodynamics (QCD)~\cite{marciano1978quantum,greiner2007quantum}, the SU(3) gauge theory of the strong interaction, confines quarks into hadrons, drives chiral symmetry breaking, and generates the vast majority of the visible mass in the universe. These phenomena are intrinsically non-perturbative and emerge from the strong-coupling infrared dynamics of the theory in a way that cannot be accessed order by order in the coupling constant. Understanding how confinement and mass generation arise from the underlying gauge dynamics remains one of the deepest open problems in theoretical physics more than fifty years after the formulation of QCD.

Lattice gauge theory~\cite{kogut1979introduction,kogut1983lattice,dalmonte2016lattice,rothe2012lattice,wilson1974confinement} has provided the most successful first-principles approach to this problem. Monte Carlo sampling of the lattice path integral has yielded precision calculations of the hadron spectrum~\cite{durr2008ab}, the QCD equation of state at finite temperature~\cite{borsanyi2014full,bazavov2014equation}, and a wide range of hadronic matrix elements relevant to flavour physics and nuclear structure~\cite{aoki2024flag}. Yet there are fundamental questions that classical lattice methods cannot address~\cite{carmen2020}. Real-time dynamics~\cite{berges2007lattice}, transport coefficients~\cite{meyer2011transport}, and the phase structure of QCD at finite baryon density all involve physics that is either inaccessible or severely limited by the sign problem~\cite{troyer2005computational} and the Euclidean signature of the lattice formulation. These are not merely technical inconveniences; they represent a hard boundary on what can be learned from classical computation alone.

Quantum simulation~\cite{georgescu2014quantum,daley2022practical,buluta2009quantum,trabesinger2012quantum,altman2021quantum} has emerged as a potentially transformative tool for circumventing these limitations. A quantum device can in principle represent the real-time dynamics of a quantum field theory directly in Minkowski signature, without analytic continuation and without a sign problem~\cite{feynman2018simulating,jordan2012quantum}. Rapid experimental progress in programmable quantum hardware has brought small-scale demonstrations of lattice gauge theory dynamics within reach~\cite{martinez2016,yang2020observation,mildenberger2025confinement}, and theoretical work has established concrete frameworks for encoding gauge theories into qubit systems~\cite{zohar2016quantum,kogut1975hamiltonian}. SU(2) lattice gauge theory, the simplest non-abelian gauge theory exhibiting both confinement and asymptotic freedom, provides the natural starting point for this program. Success in this setting would establish the viability of quantum simulation as a tool for non-perturbative gauge theory and would lay the groundwork for eventual application to the full complexity of QCD that remains beyond classical reach.

A central obstacle to realizing this program is the preservation of gauge invariance on noisy quantum hardware~\cite{pardo2023,stryker2020,kaplan2020,halimeh2020,halimeh2021,vandamme2023,tran2021,stryker2019,raychowdhury2020,rajput2023,carena2024,martinez2016,nguyen2022}. In the Hamiltonian formulation of lattice gauge theory, the physical Hilbert space is defined by a local constraint at every lattice vertex. This constraint is known as a Gauss law and takes the form of a non-abelian generalization of the requirement that electric flux be divergenceless in the absence of charges. On a classical computer, this constraint is either built into the representation or enforced algebraically, but a quantum processor has no such protection~\cite{kogut1979introduction,halimeh2020,kuhn2014quantum}. Gate errors and decoherence drive the state out of the gauge-invariant subspace, and once a gauge violation is introduced it can propagate and accumulate over subsequent time steps. For abelian gauge theories such as U(1), the Gauss law reduces to a set of commuting diagonal constraints that can be monitored with relatively simple measurements~\cite{zache2018quantum,yang2020observation,klco2018quantum}. In the non-abelian case, the constraint at each vertex involves non-commuting generators, and the gauge-invariant subspace has a richer structure that cannot be diagnosed by a single diagonal measurement. Maintaining gauge invariance in a non-abelian simulation is therefore a distinct challenge that must be solved before quantum advantage for gauge theory dynamics can be credibly pursued.

In this paper we present a protocol for actively suppressing Gauss law violations in quantum simulations of SU(2) lattice gauge theory. Mid-circuit measurements extract a syndrome $(J,M,N)$ characterizing the gauge-violation sector at each vertex, and a conditional recovery operation maps the state back to the physical subspace through an iterative sweep we call gauge cooling. The protocol refines the symmetry-based codes introduced in~\cite{bradshaw2025symmetry} by replacing the isotypic syndrome extraction with a finer-grained measurement that supplies maximal syndrome information for the recovery step. We then carry out a full Knill--Laflamme analysis at coordination-four vertices, where the singlet multiplicity is non-trivial. We prove that every single-qubit Pauli error is detected by the gauge syndrome, and we show by direct construction that the Knill--Laflamme conditions for exact recovery fail; distinct edges produce the same gauge syndrome but act distinguishably on the multiplicity space. The residual physical-subspace errors nevertheless carry a structured Pauli decomposition with vanishing $Y$ component, which suggests compatibility with concatenation by a CSS stabilizer code~\cite{gottesman1997stabilizer,bradshaw2025introduction}. We demonstrate the protocol numerically on a single-plaquette simulation of the Kogut--Susskind Hamiltonian truncated to the spin-$1/2$ representation under depolarizing and amplitude damping noise at rates representative of current superconducting hardware.

The remainder of the paper is organized as follows. In section~\ref{sec:lgt}, we review the Hamiltonian formulation of lattice gauge theory, fix conventions, and identify the obstruction to syndrome-based correction in the non-abelian case. In section~\ref{sec:syndrome}, we construct the syndrome extraction circuit using the group quantum Fourier transform and analyze its action on the data register. Section~\ref{sec:cooling} introduces the recovery operation and the iterative sweep procedure. In section~\ref{sec:kl}, we carry out the Knill--Laflamme analysis at coordination-four vertices, prove the detection theorem, exhibit the explicit failure of the recovery conditions, and characterize the residual errors. Section~\ref{sec:numerics} reports the single-plaquette numerical demonstration, including details of the Hamiltonian construction, Trotterization, and noise models. We conclude in section~\ref{sec:conclusion}. Three appendices contain the technical proofs supporting the main text; appendix~\ref{app:tdesign} establishes the unitary $t$-design condition under which the truncated syndrome extraction is exact, appendix~\ref{app:qft} treats the truncation of the group quantum Fourier transform, and appendix~\ref{app:recovery} describes the explicit construction of the recovery operator on the single-plaquette geometry.

\section{Lattice gauge theory and the Gauss law constraint}
\label{sec:lgt}

We begin with the Hamiltonian formulation of lattice gauge theory. Let $\Lambda = (V, E)$ be a finite spatial lattice with vertex set $V$ and oriented edge set $E$, and let $G$ be a compact Lie group. To each edge $e \in E$ we associate the Hilbert space
\begin{equation}
    \mathcal{H}_e = L^2(G),
\end{equation}
the space of square-integrable functions on $G$ with respect to the Haar measure~\cite{diestel2014joys,mele2024introduction}. Matter degrees of freedom may optionally reside on vertices, in which case each $v \in V$ carries a finite-dimensional Hilbert space $\mathcal{H}^{\mathrm{matt}}_v$ with a unitary representation $\rho_v \colon G \to U(\mathcal{H}^{\mathrm{matt}}_v)$. The full Hilbert space is
\begin{equation}
    \mathcal{H} = \bigotimes_{e \in E} \mathcal{H}_e
    \;\otimes\; \bigotimes_{v \in V} \mathcal{H}^{\mathrm{matt}}_v,
\end{equation}
and the local gauge group is the direct product
\begin{equation}
    \mathcal{G} = \prod_{v \in V} G_v \cong G^{|V|},
\end{equation}
where the subscript $v$ indexes the copies of $G$.

An element $\mathbf{h} = (h_v)_{v \in V} \in \mathcal{G}$ acts on an oriented edge $e$ from vertex $v$ to vertex $w$ via
\begin{equation}
    U_e(\mathbf{h}) = L_{h_v}\, R_{h_w},
\end{equation}
where $L_{h_v}$ and $R_{h_w}$ are the commuting left and right regular representations defined by
\begin{equation}
    (L_{h_v} f)(g) = f(h_v^{-1}\, g),
    \qquad
    (R_{h_w} f)(g) = f(g\, h_w).
\end{equation}
The element $\mathbf{h}$ also acts on matter via $U_v(\mathbf{h}) = \rho_v(h_v)$, and taking the tensor product over all edges and vertices yields a unitary representation $U \colon \mathcal{G} \to U(\mathcal{H})$. A state $\ket{\Psi} \in \mathcal{H}$ is \emph{physical} if it is invariant under all local gauge transformations,
\begin{equation}\label{eq:gauss_law}
    U(\mathbf{h})\ket{\Psi} = \ket{\Psi}
    \qquad \forall \, \mathbf{h} \in \mathcal{G}.
\end{equation}
The physical subspace $\mathcal{H}_{\mathrm{phys}} \subset \mathcal{H}$ is therefore the $\mathcal{G}$-invariant subspace, which in the language of representation theory is the trivial isotypic component of $\mathcal{H}$ under the action of $\mathcal{G}$.

The projector onto the physical subspace is given by the group average
\begin{equation}\label{eq:phys_projector}
    \Pi_{\mathrm{phys}} = \int_{\mathcal{G}} U(\mathbf{h})\,
    d\mathbf{h},
\end{equation}
where $d\mathbf{h}$ is the normalized Haar measure on $\mathcal{G}$. Since $\mathcal{G}$ factorizes over vertices, this integral factorizes as well,
\begin{equation}\label{eq:factorized_projector}
    \Pi_{\mathrm{phys}} = \prod_{v \in V} \Pi^{(v)}_0,
    \qquad
    \Pi^{(v)}_0 = \int_G U^{(v)}(h_v)\, dh_v,
\end{equation}
where $U^{(v)}(h_v) := U(1, \ldots, 1, h_v, 1, \ldots, 1)$ denotes the gauge action with $h_v$ non-trivial at vertex $v$ only. A projective measurement of $\{\Pi^{(v)}_0,\, \id - \Pi^{(v)}_0\}$ at each vertex therefore detects gauge-violating errors. A method for performing this measurement when $G$ is a finite group can be found in~\cite{laborde2023testing,bradshaw2023cycle,laborde2024projectors}, and a circuit implementation is shown in figure~\ref{fig:gbose_circuit}. This information alone, however, is insufficient for active error correction beyond post-selection on the physical outcome, which is computationally expensive~\cite{lamm2020suppressing,knill2005quantum}. To perform an active recovery operation one requires a syndrome that identifies the type of error that occurred without destroying the information in the state of the system. We now describe how to make such a measurement when $G = \mathrm{SU}(2)$.

\begin{figure}[t]
\centering
\begin{tikzpicture}
\begin{yquant}
  qubit {$\ket{0}$} anc;
  qubit {$\ket{\psi}$} data;
  box {$F_G$} anc;
  box {$U^{(v)}$} data | anc;
  box {$F_G^\dagger$} anc;
  measure anc;
\end{yquant}
\end{tikzpicture}
\caption{Circuit for the Gauss law test at a single vertex $v$. The ancillary register is prepared in the uniform superposition $\ket{+_G} = |G|^{-1/2} \sum_g \ket{g}$ by the unitary $F_G$. The controlled gauge action applies $U^{(v)}(g)$ to the data register conditioned on the ancilla state $\ket{g}$. The inverse $F_G^\dagger$ is applied and the ancilla is measured; the outcome corresponding to $\ket{0}$ is the ``accept'' outcome, with acceptance probability $\bra{\psi} \Pi^{(v)}_0 \ket{\psi}$.}
\label{fig:gbose_circuit}
\end{figure}
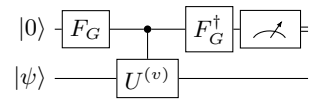

By the Peter--Weyl theorem, the link Hilbert spaces $\mathcal{H}_e = L^2(\mathrm{SU}(2))$ decompose as
\begin{equation}\label{eq:pw_edge}
    \mathcal{H}_e \cong \bigoplus_{j \in \frac{1}{2}\mathbb{N}_0}
    V_j \otimes V_j^*,
\end{equation}
where $V_j$ is the carrier space of the spin-$j$ irreducible representation (irrep) $\pi_j$ with dimension $d_j = 2j+1$ and $\mathbb{N}_0$ denotes the non-negative integers. The normalized Wigner matrix elements $\sqrt{d_j}\, D^{(j)}_{mn}(g) = \sqrt{2j+1}\, [\pi_j(g)]_{mn}$ with $-j \le m, n \le j$ form a complete orthonormal basis for $\mathcal{H}_e$. In practice the sum over $j$ is truncated at $j_{\mathrm{max}}$, reducing each edge to a finite-dimensional Hilbert space of dimension $\sum_{j=0}^{j_{\mathrm{max}}} (2j+1)^2$. Each summand $V_j \otimes V_j^*$ in~\eqref{eq:pw_edge} carries a physical interpretation; for an oriented edge $e$ the factor $V_{j_e}$ transforms under gauge transformations at the source vertex and $V_{j_e}^*$ transforms at the target.

Regrouping the tensor product over edges by vertex and applying the Clebsch--Gordan decomposition (see appendix~\ref{app:tdesign} for the full derivation), we obtain the vertex Hilbert space at each vertex $v$,
\begin{equation}\label{eq:vertex_decomp}
    \mathcal{H}_v\bigl(\{j_e\}\bigr) \cong \bigoplus_{J_v}
    V_{J_v} \otimes \mathbb{C}^{\mu_{J_v}},
\end{equation}
where $J_v$ labels the total angular momentum at vertex $v$ and $\mu_{J_v}$ is the corresponding Clebsch--Gordan multiplicity. Fixing a standard basis $\{\ket{M}\}_{M=-J}^{J}$ for $V_J$, we define the subspace
\begin{equation}\label{eq:WMJ}
    \mathcal{W}_M^J = \mathrm{span}\bigl\{\ket{M} \otimes \ket{\alpha}
    : \alpha = 1, \ldots, \mu_J\bigr\}
    \subset V_J \otimes \mathbb{C}^{\mu_J},
\end{equation}
which is the $\mu_J$-dimensional subspace of the $J$-isotypic component with magnetic quantum number $M$. The Gauss law constraint requires that the physical subspace correspond to the singlet sector $J_v = 0$ at every vertex. The quantum numbers $J_v$ and the associated magnetic projections together characterize gauge violations at vertex $v$, and it is the extraction of these quantum numbers by mid-circuit measurement that enables the active correction protocol.

\section{Syndrome extraction by group quantum Fourier transform}
\label{sec:syndrome}

The syndrome extraction protocol uses an ancillary register coupled to the data register through the controlled gauge action, followed by a group quantum Fourier transform and a basis measurement. The circuit is shown in figure~\ref{fig:wigner_circuit} and proceeds in four steps.

\begin{figure}[t]
\centering
\begin{tikzpicture}
\begin{yquant}
  qubit {$\ket{0}$} anc;
  qubit {$\ket{\psi}$} data;
  box {$F_{\mathrm{SU(2)}}$} anc;
  box {$U^{(v)}(g)$} data | anc;
  box {$\mathrm{QFT}_{\mathrm{SU(2)}}$} anc;
  measure anc;
\end{yquant}
\end{tikzpicture}
\caption{Syndrome extraction at vertex $v$. The circuit is identical to that in figure~\ref{fig:gbose_circuit} except that the inverse preparation $F_{\mathrm{SU(2)}}^\dagger$ is replaced by the group quantum Fourier transform $\mathrm{QFT}_{\mathrm{SU(2)}}$. Measuring the ancillary register in the Wigner basis yields the syndrome $(J,M,N)$. The operator applied to the data register selects the component in $\mathcal{W}_N^J$ and maps it into $\mathcal{W}_M^J$, leaving the multiplicity degrees of freedom untouched.}
\label{fig:wigner_circuit}
\end{figure}
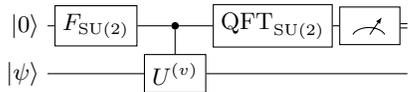

\subsection{Step 1, prepare the ancillary register}

Ideally one would initialize an ancillary register in a uniform superposition over the basis for the group algebra generated by $\mathrm{SU}(2)$, but finite-dimensional hardware forces a truncation. We show in appendix~\ref{app:tdesign} that preparing a uniform superposition over a unitary $t$-design is sufficient. Such a superposition has the form
\begin{equation}\label{eq:t_design_super}
    \frac{1}{\sqrt{n_t}} \sum_{i=1}^{n_t} \ket{g_i}
\end{equation}
for some $n_t$ group elements $g_i \in \mathrm{SU}(2)$. The required design strength $t$ depends on the truncation $j_{\mathrm{cut}}$ of the group QFT and on the coordination of the vertex; the precise condition is established in appendix~\ref{app:tdesign}. We may encode the group basis states however we like, so long as the controlled-$U^{(v)}(g_i)$ operation in figure~\ref{fig:wigner_circuit} is triggered when the ancillary register is in the state $\ket{g_i}$ and the final QFT is performed according to the same basis assignment.

\subsection{Step 2, apply the controlled gauge action}

Perform the operation
\begin{equation}\label{eq:controlled_gauge}
    \sum_{i=1}^{n_t} \ket{g_i}\!\bra{g_i}
    \otimes U^{(v)}(g_i),
\end{equation}
which applies the gauge action $U^{(v)}(g_i)$ at vertex $v$ to the data register, conditioned on the ancillary register being in state $\ket{g_i}$. After this step the joint state of the ancilla and data registers for an input state $\ket{\psi}$ is
\begin{equation}\label{eq:joint_state}
    \frac{1}{\sqrt{n_t}}
    \sum_{i=1}^{n_t}
    \ket{g_i} \otimes U^{(v)}(g_i) \ket{\psi}.
\end{equation}

\subsection{Step 3, apply the group quantum Fourier transform}

The quantum Fourier transform admits a generalization to a group quantum Fourier transform, where the ordinary transform is recovered by choosing a cyclic group. For $\mathrm{SU}(2)$, this operation is defined by
\begin{equation}\label{eq:qft}
    \qft_{\mathrm{SU(2)}} \ket{g}
    = \sum_{j \in \frac{1}{2}\mathbb{N}_0} \sqrt{2j+1}
    \sum_{m, n = -j}^{j}
    [\overline{\pi_j(g)}]_{m,n}\,
    \ket{j, m, n},
\end{equation}
which acts on an infinite-dimensional Hilbert space. We again encounter the finite-dimensional hardware limitation, which we alleviate by truncating the QFT along the $j$ axis. Defining a cutoff $j_{\mathrm{cut}}$, the truncated transform acts as
\begin{equation}\label{eq:qft_truncated}
    \qft_{\mathrm{SU(2)}}^{\le j_{\mathrm{cut}}} \ket{g}
    = \sum_{j=0}^{j_{\mathrm{cut}}} \sqrt{\frac{2j+1}{n_t}}
    \sum_{m, n = -j}^{j}
    [\overline{\pi_j(g)}]_{m,n}\,
    \ket{j, m, n},
\end{equation}
where $j$ ranges over half-integer steps and the factor of $1/\sqrt{n_t}$ provides the appropriate normalisation. Justification of this truncation is deferred to appendix~\ref{app:qft}, where we prove that the truncated transform is an isometry from the $n_t$-dimensional ancilla register into the subspace spanned by $\{\ket{j,m,n} : j \le j_{\mathrm{cut}},\, -j \le m, n \le j\}$ and is embedded into a unitary on the full ancilla register by extension to an orthonormal basis. The relationship between the group algebra basis $\{\ket{g}\}$ and the Wigner basis $\{\ket{j,m,n}\}$ is discussed in appendix~\ref{app:bases}.

\subsection{Step 4, measure the ancillary register}

Measure the ancillary register in the $\ket{j,m,n}$ basis, obtaining outcome $(J, M, N)$. To determine the post-measurement state of the data register, we substitute the output of step 2 into step 3 and project onto the measurement outcome. After the group QFT, the joint state of the ancilla and data registers is
\begin{equation}\label{eq:after_qft}
    \sum_{j, m, n} \ket{j, m, n} \otimes T_{mn}^{(j)} \ket{\psi},
\end{equation}
where
\begin{equation}\label{eq:T_operator}
    T_{mn}^{(j)} = \frac{\sqrt{2j+1}}{n_t}
    \sum_{i=1}^{n_t} [\overline{\pi_j(g_i)}]_{mn}\,
    U^{(v)}(g_i)
\end{equation}
is the discrete approximation to the operator
\begin{equation}\label{eq:T_haar}
    \sqrt{2j+1} \int_{\mathrm{SU(2)}}
    \overline{D^{(j)}_{mn}(g)}\, U^{(v)}(g)\, dg,
\end{equation}
which, by the Peter--Weyl orthogonality relations, acts within the $j$-isotypic component of $\mathcal{H}_v$ as
\begin{equation}\label{eq:T_form}
    T_{mn}^{(j)} = \frac{1}{\sqrt{2j+1}}\,
    \ket{m}\!\bra{n} \otimes \id_{\mu_j}.
\end{equation}
Upon obtaining outcome $(J, M, N)$, the data register is left in the unnormalized state $T_{MN}^{(J)} \ket{\psi}$. This operator selects the component of $\ket{\psi}$ in the subspace $\mathcal{W}_N^J$ and maps it into $\mathcal{W}_M^J$, leaving the multiplicity degrees of freedom untouched. The outcome occurs with probability
\begin{equation}\label{eq:outcome_prob}
    p(J, M, N) = \bra{\psi} \bigl(T_{MN}^{(J)}\bigr)^\dagger
    T_{MN}^{(J)} \ket{\psi}
    = \frac{1}{2J+1}\, \bra{\psi} P_N^J \ket{\psi},
\end{equation}
where $P_N^J = \ket{N}\!\bra{N} \otimes \id_{\mu_J}$ is the projector onto $\mathcal{W}_N^J$ within the $J$-isotypic component. The probability is independent of $M$, so conditioned on $J$ and $N$ the outcome $M$ is uniformly distributed over $\{-J, \ldots, J\}$. For a gauge-invariant input state, $J = M = N = 0$ with certainty and $T_{00}^{(0)}$ acts as the identity on the singlet sector.

\section{Gauge cooling}
\label{sec:cooling}

Having extracted the syndrome $(J, M, N)$ at a vertex, we now describe the recovery operation. If the measurement outcome is $(J, M, N) = (0, 0, 0)$, no gauge violation has occurred and no action is taken. Otherwise we apply a recovery unitary $R_{J, M}$ to the data register that maps the subspace $\mathcal{W}_M^J$ into $\mathcal{W}_0^0$. Concretely $R_{J, M}$ is any unitary on $\mathcal{H}_v$ satisfying
\begin{equation}\label{eq:recovery}
    R_{J, M}\, \ket{M} \otimes \ket{\alpha}
    = \ket{0} \otimes \ket{\alpha},
    \qquad
    \alpha = 1, \ldots, \min(\mu_J, \mu_0),
\end{equation}
where $\ket{M} \otimes \ket{\alpha}$ is understood as a vector in the summand $V_J \otimes \mathbb{C}^{\mu_J} \subset \mathcal{H}_v$ and $\ket{0} \otimes \ket{\alpha}$ as a vector in $V_0 \otimes \mathbb{C}^{\mu_0} \subset \mathcal{H}_v$, with the action on the remaining basis vectors chosen to complete a unitary. When $\mu_J > \mu_0$ the recovery maps only $\mu_0$ of the $\mu_J$ multiplicity states back to the singlet sector; the remaining $\mu_J - \mu_0$ states are mapped to an orthogonal subspace and their information is lost. The recovery depends on $J$ and $M$ but not on $N$; the third index records which subspace the state occupied before measurement, but after measurement the data register already resides in $\mathcal{W}_M^J$ regardless of $N$.

\begin{definition}[Gauge cooling]
\label{def:gauge_cooling}
The composition of syndrome extraction (steps 1--4 of section~\ref{sec:syndrome}) with the conditional recovery in~\eqref{eq:recovery} constitutes the \emph{gauge cooling} operation at a single vertex. \emph{Iterative gauge cooling} on the lattice $\Lambda$ refers to the procedure in which gauge cooling is applied at each vertex in sequence and the full sweep is repeated until a prescribed convergence criterion on the gauge-invariant overlap is met or a maximum number of sweeps is reached.
\end{definition}

The iterative structure is necessary because the recovery at one vertex modifies the edge spins on the edges incident to that vertex, which in turn can introduce gauge violations at neighboring vertices that share those edges. We define the gauge-invariant overlap
\begin{equation}\label{eq:gi_overlap}
    \mathcal{F}_{\mathrm{GI}}(\rho) = \frac{1}{|V|}
    \sum_{v \in V} \tr\!\left(\Pi_0^{(v)}\, \rho\right),
\end{equation}
where $\Pi_0^{(v)}$ is the projector onto the $J_v = 0$ sector at vertex $v$. A state that is exactly gauge-invariant at every vertex has $\mathcal{F}_{\mathrm{GI}} = 1$. In the simulations of section~\ref{sec:numerics} the iterative sweep is terminated when $\mathcal{F}_{\mathrm{GI}} > 1 - \epsilon$ for a tolerance $\epsilon = 10^{-5}$ or after a maximum of ten sweeps.

The convergence of iterative gauge cooling is geometric in the single-plaquette geometry studied numerically in section~\ref{sec:numerics}, with a contraction factor of approximately $0.45$ per sweep. Each sweep is a completely positive trace-preserving map on the full density matrix, and the gauge-invariant subspace is a fixed point. The gauge-violating component is reduced by each sweep at a rate governed by the spectral gap of the sweep map restricted to the orthogonal complement of the physical subspace. The number of sweeps required for convergence on larger lattices is expected to depend on the lattice size and connectivity, since information about gauge violations must propagate across shared edges. This is analogous to the convergence of iterative gauge fixing in classical lattice gauge theory, where Landau gauge is imposed by sweeping over vertices and maximizing a local functional, with the number of sweeps growing polynomially with the linear lattice size~\cite{mandula1987gluon,cucchieri2003critical}. A systematic study of this scaling is left to future work.

The gauge cooling operation admits a Kraus representation that we record for use in section~\ref{sec:kl}. After the spectator-preserving basis assignment described in appendix~\ref{app:recovery}, the Kraus operators take the form
\begin{equation}\label{eq:kraus_gc}
    K_{J, N} = \sum_{\alpha=1}^{\min(\mu_J,\, \mu_0)}
    \ket{0, 0, \alpha}\!\bra{J, N, \alpha},
\end{equation}
where the sum runs over the multiplicity index. The full gauge cooling channel at a single vertex is therefore
\begin{equation}\label{eq:gc_channel}
    \mathcal{E}(\rho) = \sum_{J, N} K_{J, N}\, \rho\,
    K_{J, N}^\dagger,
\end{equation}
which is trace-preserving; $\sum_{J, N} K_{J, N}^\dagger K_{J, N} = \sum_{J, N} P_N^J = \id$.

\section{Knill--Laflamme analysis at coordination-four vertices}
\label{sec:kl}

The single-plaquette demonstration in section~\ref{sec:numerics} uses coordination-two vertices, where the singlet multiplicity within each edge-spin block is $\mu_0 = 1$ and gauge cooling constitutes exact error correction within each block. On larger lattices with square or higher topology, vertices have coordination number four or greater and the singlet multiplicity grows beyond one. The physical subspace at such a vertex carries non-trivial degrees of freedom, and the question of whether gauge cooling can protect those degrees of freedom becomes substantive. We answer this question for the simplest non-trivial case; a single coordination-four vertex with four spin-$1/2$ edges. The analysis establishes three structural results; every single-qubit Pauli error is detected (theorem~\ref{thm:detection}), the Knill--Laflamme conditions for exact recovery fail (proposition~\ref{prop:kl_failure}), and the residual physical-subspace errors carry a structured Pauli decomposition.

\subsection{Setup}
\label{sec:kl_setup}

Consider a single vertex $v$ with four incident edges, each carrying spin $j_{\mathrm{max}} = 1/2$. The gauge transformation at $v$ acts on the tensor product of the four spin-$1/2$ representations contributed by these edges. Since we are interested only in the indices that transform under the gauge group at $v$, the relevant Hilbert space is $(V_{1/2})^{\otimes 4} \cong (\mathbb{C}^2)^{\otimes 4}$, which is sixteen-dimensional.

To decompose this space into sectors of definite total angular momentum $J$, we choose the coupling scheme $(12)(34) \to J$; first couple edges 1 and 2 to an intermediate spin $j_{12} \in \{0, 1\}$, then couple edges 3 and 4 to $j_{34} \in \{0, 1\}$, and finally couple $j_{12}$ and $j_{34}$ to the total $J$. The result is the standard Clebsch--Gordan decomposition
\begin{equation}\label{eq:coord4_decomp}
    \left(\tfrac{1}{2}\right)^{\otimes 4}
    \cong 0^2 \oplus 1^3 \oplus 2^1,
\end{equation}
where the superscripts denote multiplicities. This decomposition is independent of the coupling order; a different scheme such as $(13)(24) \to J$ yields the same multiplicities but labels the states differently.

The singlet sector $J = 0$ has multiplicity $\mu_0 = 2$, so the physical subspace at this vertex is two-dimensional. The two singlet states correspond to the intermediate spin assignments $(j_{12}, j_{34}) = (0, 0)$, in which both pairs form singlets, and $(j_{12}, j_{34}) = (1, 1)$, in which both pairs form triplets that then couple to a total singlet. The $J = 1$ sector has multiplicity $\mu_1 = 3$, with the three channels labeled by $(j_{12}, j_{34}) = (0, 1),\, (1, 0),\, (1, 1)$. Together with the $2J + 1 = 3$ magnetic substates this gives $3 \times 3 = 9$ states. The $J = 2$ sector has $\mu_2 = 1$, giving $2J + 1 = 5$ states. The total is $2 + 9 + 5 = 16$, consistent with the dimension of $(\mathbb{C}^2)^{\otimes 4}$.

\subsection{Detection of single-qubit errors}
\label{sec:detection}

We now analyze how single-qubit errors transform under the gauge group. Any operator on a single spin-$1/2$ edge decomposes into irreducible spherical tensor components under the adjoint action of the gauge group. The space of all $2 \times 2$ matrices is four-dimensional and decomposes as $V_{1/2} \otimes V_{1/2} \cong V_0 \oplus V_1$, that is, a one-dimensional scalar ($J = 0$) piece and a three-dimensional vector ($J = 1$) piece. The scalar piece is the identity matrix $\id_2$, the unique operator commuting with all SU(2) transformations by Schur's lemma. The vector piece is the three-dimensional space of traceless Hermitian matrices spanned by the Pauli matrices $\{X, Y, Z\}$.

Since the Pauli matrices are traceless, a single-qubit Pauli error $P \in \{X, Y, Z\}$ on edge $k$ is purely $J = 1$ and maps the singlet sector entirely into the $J = 1$ sector. The $J = 1$ sector has three magnetic substates $M = -1, 0, +1$. To determine which substate the error produces, we re-express the Pauli matrices in the spherical tensor basis whose elements have definite magnetic quantum number $q$,
\begin{equation}\label{eq:spherical_pauli}
    O^{(1)}_0 = Z,
    \qquad
    O^{(1)}_{\pm 1} = \mp\frac{1}{\sqrt{2}}\bigl(X \pm i Y\bigr).
\end{equation}
When a spherical tensor with component $q$ acts on a singlet $M = 0$, the selection rule $\Delta M = q$ determines the magnetic quantum number of the output state. The operator $Z$, being the $q = 0$ component, maps the singlet sector into $M = 0$ of $J = 1$; the operators $X$ and $Y$, being linear combinations of the $q = \pm 1$ components, map the singlet sector into the span of the $M = \pm 1$ subspaces of $J = 1$. This decomposition by $M$ sector is essential for the Knill--Laflamme analysis that follows, because errors landing in different $M$ sectors produce different syndromes and are handled by different recovery operations.

No single-qubit error can reach $J = 2$, because the operator space on a single spin-$1/2$ decomposes as $0 \oplus 1$ and therefore cannot induce $J = 2$ transitions. Two-qubit errors acting on two edges simultaneously can reach $J = 2$.

\begin{theorem}[Detection]
\label{thm:detection}
Let $\Pi_0$ denote the projector onto the singlet sector at a coordination-four vertex with four spin-$1/2$ edges. For every Pauli operator $P \in \{X, Y, Z\}$ and every edge $k \in \{0, 1, 2, 3\}$,
\begin{equation}\label{eq:detection}
    \Pi_0\, \bigl(P_k \otimes \id_{\bar k}\bigr)\, \Pi_0 = 0,
\end{equation}
where $\bar k$ denotes the complement of edge $k$.
\end{theorem}

\begin{proof}
By the Wigner--Eckart theorem, the matrix element of a rank-$J$ spherical tensor operator $T^{(J)}_q$ between states of total angular momentum $J_a$ and $J_b$ factorizes as a Clebsch--Gordan coefficient times a reduced matrix element,
\begin{equation}
    \langle J_b, M_b \,|\, T^{(J)}_q \,|\, J_a, M_a \rangle
    = \langle J_a, M_a;\, J, q \,|\, J_b, M_b \rangle\,
      \langle J_b \,\|\, T^{(J)} \,\|\, J_a \rangle.
\end{equation}
The Clebsch--Gordan coefficient vanishes unless $J_b$ appears in the decomposition of $V_{J_a} \otimes V_J$, which requires the triangle condition $|J_a - J| \le J_b \le J_a + J$. Since Pauli operators are traceless and therefore purely rank $J = 1$ as established in~\eqref{eq:spherical_pauli}, and since the singlet sector has $J_a = J_b = 0$, the triangle condition requires $|0 - 1| \le 0 \le 0 + 1$, that is, $1 \le 0$, which is false. The Clebsch--Gordan coefficient therefore vanishes for every choice of magnetic quantum numbers, and equation~\eqref{eq:detection} follows.
\end{proof}

The physical consequence of Theorem~\ref{thm:detection} is that every single-qubit Pauli error necessarily produces a gauge violation $J \ne 0$ and is therefore detected by the syndrome extraction. This is a strong property; the gauge code detects all single-qubit Pauli errors at a coordination-four vertex with spin-$1/2$ edges. However, as we will now show, this does not mean that every such error can be corrected.

\subsection{Failure of the Knill--Laflamme conditions}
\label{subsec:kl_failure}

We have established that every single-qubit Pauli error on edge $k$ is detected by
the gauge syndrome. We now ask the stronger question of whether the syndrome
contains enough information to identify the error and recover from it. The
answer is no, and the obstruction lies in the multiplicity space of the singlet
sector.

\subsubsection{Decomposing the error into a syndrome part and a multiplicity part}

A Pauli error on edge $k$, restricted to its action on the singlet sector,
admits a Wigner--Eckart factorisation
\begin{equation}\label{eq:we_factorization}
    P_k\bigl|_{V_0 \otimes \mathbb{C}^{\mu_0}}
    \;=\;
    \sum_{q} O^{(1)}_q \otimes A_k^{(q)},
\end{equation}
where $O^{(1)}_q$ is the spherical tensor component carrying the gauge quantum
numbers and $A_k^{(q)}$ is a linear map between multiplicity spaces. The
spherical operator $O^{(1)}_q$ determines which $J=1$ magnetic sector receives
the output, namely $M = q$ when acting on the singlet $M=0$ state. The
multiplicity-space factor $A_k^{(q)}$ is the reduced matrix element. It encodes
which edge the error acted on and is independent of the angular labels.

For the coordination-four vertex, the singlet multiplicity is $\mu_0 = 2$ and
the $J=1$ multiplicity is $\mu_1 = 3$, so each $A_k^{(q)}$ is a $3 \times 2$
matrix. We adopt the basis for the multiplicity spaces induced by the
$(12)(34)$ coupling scheme. The two singlets are
\begin{equation}\label{eq:singlets}
    |\alpha_1\rangle = |0\rangle_{12} \otimes |0\rangle_{34},
    \qquad
    |\alpha_2\rangle = \frac{1}{\sqrt{3}}\sum_{m=-1}^{1} (-1)^{1-m}\,
        |1, m\rangle_{12} \otimes |1, -m\rangle_{34},
\end{equation}
where the subscripts label the qubit pair and $|j, m\rangle_{ab}$ denotes the
spin-$j$ state with magnetic quantum number $m$ in the angular momentum coupling
of qubits $a$ and $b$. The first singlet pairs the two edge pairs into local
singlets. The second couples each pair into a local triplet and then combines
the two triplets into a total singlet via the standard $1 \otimes 1 \to 0$
Clebsch--Gordan coefficients.

\subsubsection{Computing $A_k^{(q)}$ for $Z$ errors}

We compute $A_k^{(0)}$ explicitly for each $k \in \{0, 1, 2, 3\}$ when the error
is $Z_k$. In this case, the $q=0$ component is the only relevant one because $Z = O^{(1)}_0$. The result of $Z_k$ acting on each singlet is read off by tracking the action of $Z$ on the constituent computational basis states.

For $Z_0$ acting on $|\alpha_1\rangle$, the qubit-$0$ slot of the
$|0\rangle_{12} = (|01\rangle - |10\rangle)/\sqrt{2}$ factor picks up a sign
according to its computational basis value, converting the singlet into the
triplet $|1, 0\rangle_{12} = (|01\rangle + |10\rangle)/\sqrt{2}$. The
$|0\rangle_{34}$ factor is untouched. We obtain
\begin{equation}\label{eq:z0_alpha1}
    Z_0 |\alpha_1\rangle = |1, 0\rangle_{12} \otimes |0\rangle_{34}.
\end{equation}
This state lies in the $(j_{12}, j_{34}) = (1, 0)$ multiplicity channel of the
$J=1, M=0$ sector. For $Z_1$ the same calculation gives the opposite sign
because the surviving terms exchange roles under interchange of the two qubits,
\begin{equation}\label{eq:z1_alpha1}
    Z_1 |\alpha_1\rangle = -|1, 0\rangle_{12} \otimes |0\rangle_{34}.
\end{equation}

The action on $|\alpha_2\rangle$ is more involved. The operator $Z_0$, for example, leaves $|1, 1\rangle_{12}$ invariant, introduces a sign to $|1, -1\rangle_{12}$, and converts $|1, 0\rangle_{12}$ into $|0\rangle_{12}$. Tracking each term in \eqref{eq:singlets}, we find
\begin{align}
    Z_0 |\alpha_2\rangle &= \frac{1}{\sqrt{3}}\Bigl[
        |1, 1\rangle_{12} |1, -1\rangle_{34}
        - |0\rangle_{12} |1, 0\rangle_{34}
        - |1, -1\rangle_{12} |1, 1\rangle_{34}
    \Bigr], \label{eq:z0_alpha2}\\
    Z_1 |\alpha_2\rangle &= \frac{1}{\sqrt{3}}\Bigl[
        |1, 1\rangle_{12} |1, -1\rangle_{34}
        + |0\rangle_{12} |1, 0\rangle_{34}
        - |1, -1\rangle_{12} |1, 1\rangle_{34}
    \Bigr]. \label{eq:z1_alpha2}
\end{align}
The two states agree on the $|1, \pm 1\rangle_{12}$ pieces and differ in sign on
the $|0\rangle_{12}|1, 0\rangle_{34}$ piece. The reason is that $Z_0$ and $Z_1$
agree on the symmetric triplet states $|1, \pm 1\rangle_{12}$ but disagree on
the antisymmetric singlet $|0\rangle_{12}$, since interchange of the two qubits
acts trivially on triplets and as $-1$ on the singlet.

\subsubsection{The Knill--Laflamme product}

The Knill--Laflamme conditions, restricted to a fixed syndrome sector
$(J, M, N)$, require that for every pair of correctable errors $E_a, E_b$ with
multiplicity-space factors $A_a, A_b$ in that sector,
\begin{equation}\label{eq:kl_condition}
    A_a^\dagger A_b \;\propto\; \mathbf{1}_{\mu_0}.
\end{equation}
Proportionality to the identity on the singlet multiplicity space is what
allows a syndrome-conditional recovery to undo all errors in the sector with a
single unitary, since the recovery has access only to the syndrome and cannot
adapt to the multiplicity-space action.

We test this condition for $Z$ errors on edges $0$ and $1$, which produce the
same syndrome $(J, M, N) = (1, 0, 0)$. The matrix entries of $A_0^\dagger A_1$
are the inner products $\langle Z_0 \alpha_i | Z_1 \alpha_j \rangle$. The diagonal entry $(1,1)$ is computed from \eqref{eq:z0_alpha1} and
\eqref{eq:z1_alpha1},
\begin{equation}\label{eq:kl_11}
    \langle Z_0 \alpha_1 | Z_1 \alpha_1 \rangle = -\langle 1, 0 | 1, 0\rangle_{12}
    \langle 0 | 0 \rangle_{34} = -1.
\end{equation}
The diagonal entry $(2,2)$ is computed from \eqref{eq:z0_alpha2} and
\eqref{eq:z1_alpha2}. The two $|1, \pm 1\rangle$ pieces contribute $+1/3$ each,
the $|0\rangle_{12}|1, 0\rangle_{34}$ piece contributes $-1/3$ from the sign
mismatch, and the cross terms vanish by orthogonality, producing
\begin{equation}\label{eq:kl_22}
    \langle Z_0 \alpha_2 | Z_1 \alpha_2 \rangle = \tfrac{1}{3} + \tfrac{1}{3} - \tfrac{1}{3} = \tfrac{1}{3}.
\end{equation}
The off-diagonal entries vanish because $Z_0|\alpha_1\rangle$ has $j_{34} = 0$
while $Z_1|\alpha_2\rangle$ has $j_{34} = 1$ in every term, and states with
distinct $j_{34}$ are orthogonal. We obtain
\begin{equation}\label{eq:kl_counterexample}
    A_0^\dagger A_1 =
    \begin{pmatrix} -1 & 0 \\ 0 & \tfrac{1}{3} \end{pmatrix}.
\end{equation}
This matrix is diagonal but its eigenvalues are unequal, so it is not
proportional to $\mathbf{1}_2$.

\begin{proposition}[Failure of Knill--Laflamme]
\label{prop:kl_failure}
The Knill--Laflamme conditions for exact correction of single-qubit Pauli
errors by syndrome-based recovery fail at a coordination-four vertex with four
spin-$1/2$ edges.
\end{proposition}

\begin{proof}
Equation \eqref{eq:kl_counterexample} exhibits two errors $Z_0$ and $Z_1$ that
produce the same gauge syndrome and yet have $A_0^\dagger A_1$ not proportional
to the identity on the singlet multiplicity space. The Knill--Laflamme
condition \eqref{eq:kl_condition} is violated.
\end{proof}

\subsubsection{Structure of the failure}

The counterexample \eqref{eq:kl_counterexample} is representative of a broader
pattern. Let us summarize the structural features of $A_k^\dagger A_l$ for all
edge pairs. The diagonal terms $A_k^\dagger A_k$ are proportional to $\mathbf{1}_2$ for
every edge $k$ and every Pauli type. This follows from the fact that
$A_k^{(q)}$ has equal singular values for every $k$ and every spherical
component $q$. In the spherical tensor basis, the singular values of each
$A_k^{(q)}$ are $(1, 1)$. Cartesian errors $X$ and $Y$ split their weight
between $q = +1$ and $q = -1$ with coefficient $1/\sqrt{2}$, so when projected
onto a definite $M = \pm 1$ sector, the corresponding multiplicity-space matrix
has singular values $(1/\sqrt{2}, 1/\sqrt{2})$. The diagonal Knill--Laflamme
products $A_k^\dagger A_k$ are therefore $(1/2)\mathbf{1}_2$ in the
$M = \pm 1$ sectors and $\mathbf{1}_2$ in the $M = 0$ sector.

The off-diagonal terms $A_k^\dagger A_l$ for $k \neq l$ depend on whether the
two edges are paired in the coupling scheme. Edges $0$ and $1$ are paired
through the first-stage coupling, as are edges $2$ and $3$. For paired edges
the matrices $A_k$ and $A_l$ have nonzero entries in the same positions but
differ by sign flips in entries that involve the pair-singlet channel
$j_{12} = 0$ or $j_{34} = 0$. The reason is that the qubit-interchange action
on the pair Hilbert space is $+1$ on the triplet and $-1$ on the singlet, and
$Z_k - Z_l$ is antisymmetric under interchange. The diagonal entries of
$A_k^\dagger A_l$ pick up unequal signs from this interchange, as in
\eqref{eq:kl_counterexample}.

For cross-pair edges, for example $k=0$ and $l=2$, the matrices $A_k$ and $A_l$
have nonzero entries in different positions. The reason is that $Z_0$ can
change $j_{12}$ but not $j_{34}$, while $Z_2$ can change $j_{34}$ but not
$j_{12}$. The product $A_0^\dagger A_2$ then has off-diagonal entries in the
multiplicity basis, not just unequal diagonal entries, and this is even further
from proportionality to the identity than the paired case.

In every case the off-diagonal Knill--Laflamme products fail proportionality
to the identity. Different edges produce the same gauge syndrome but act
distinguishably on the multiplicity space. The syndrome measurement records
the angular quantum numbers of the error and discards the edge label, and the
discarded information is exactly what is needed to choose the correct recovery.
No syndrome-based recovery can perfectly undo all single-qubit Pauli errors at
this vertex.

\subsection{Residual errors after gauge cooling}
\label{sec:residual} After gauge cooling, the state is returned to the singlet sector but the recovery cannot distinguish which edge caused the error. The result is a residual logical error on the two-dimensional multiplicity space $\mathbb{C}^{\mu_0} = \mathbb{C}^2$. To characterize these residual errors we fix a reference recovery, namely the pseudoinverse of the error map from edge~$0$, and study the effective $2 \times 2$ operator $R \cdot A_k = A_0^+ A_k$ for each edge~$k$.

\begin{theorem}[Universality of residual errors]
\label{thm:universality}
At a coordination-four vertex with four spin-$1/2$ edges, the residual error $A_0^+ A_k$ on the multiplicity qubit is independent of the Pauli type ($X$, $Y$, or $Z$) and independent of the syndrome sector~$M$.
\end{theorem}

\begin{proof}
By the Wigner--Eckart theorem, the multiplicity-space matrix $A_k^{(q)}$ is the reduced matrix element of a rank-$1$ spherical tensor acting on edge~$k$, coupling $J_a = 0$ to $J_b = 1$. The reduced matrix element depends only on the edge label~$k$ and the representation labels $(J_a, J_b, J) = (0, 1, 1)$, not on the spherical component~$q$ or the magnetic quantum numbers. Since all three Pauli matrices are rank-$1$ tensors with the same representation labels, they share the same reduced matrix elements. The products $A_0^+ A_k$ therefore depend only on~$k$.
\end{proof}

By theorem~\ref{thm:universality} it suffices to evaluate the residual errors in a single syndrome sector. We compute $A_0^+ A_k$ using $Z$ errors in the $M = 0$ sector, decompose the resulting $2 \times 2$ matrices into Pauli components, and record the normalized weights $|c_P|^2 / \sum_{P'} |c_{P'}|^2$ in table~\ref{tab:residual}.

\begin{table}[h]
\centering
\caption{Residual error weights after gauge cooling at a coordination-four vertex with four spin-$1/2$ edges. Edge~$0$ is taken as the reference. By theorem~\ref{thm:universality}, the weights are independent of the Pauli type and syndrome sector.}
\label{tab:residual}
\begin{tabular}{@{}lccc@{}}
\hline
Error & $I$ weight & $X$ weight & $Z$ weight \\
\hline
$Z_0$ & $1.00$ & $0$ & $0$ \\
$Z_1$ & $0.20$ & $0$ & $0.80$ \\
$Z_2$ & $0.20$ & $0.60$ & $0.20$ \\
$Z_3$ & $0.20$ & $0.60$ & $0.20$ \\
\hline
\end{tabular}
\end{table}

The absence of $Y$ errors and the bounded weights of $X$ and $Z$ errors establish that the residual error channel on the multiplicity qubit has support only on $\{I, X, Z\}$, which is precisely the condition for compatibility with CSS-type concatenation. An $[[n, 1, d]]$ CSS code encoding the multiplicity qubit across $n$ vertices could correct these residuals independently in the $X$ and $Z$ bases, implementing a two-layer error correction scheme in which gauge cooling eliminates gauge-variant errors and a stabilizer code corrects the residual multiplicity-space distortions. We leave the construction of this concatenated scheme to future work.

\section{Single-plaquette demonstration}
\label{sec:numerics}

To validate the protocol, we simulate Trotterized time evolution of the Kogut--Susskind Hamiltonian~\cite{kogut1975hamiltonian} on a single plaquette with $j_{\mathrm{max}} = 1/2$ and gauge coupling $g^2 = 1$. After each Trotter step, independent noise is applied to every edge, followed by iterative gauge cooling at all four vertices. We compare the fidelity of the corrected state with that of the uncorrected state, each measured against the ideal noiseless evolution, under both depolarizing and amplitude damping noise models.

\subsection{Geometry and Hamiltonian}
\label{sec:geometry}

The single-plaquette geometry consists of four vertices $\{v_0, v_1, v_2, v_3\}$ and four oriented edges $\{e_0, e_1, e_2, e_3\}$ arranged so that the plaquette operator traces clockwise around the square; $e_0 : v_0 \to v_1$, $e_1 : v_1 \to v_2$, $e_2 : v_2 \to v_3$, $e_3 : v_3 \to v_0$. Each vertex has coordination two, with one outgoing and one incoming edge. With the truncation $j_{\mathrm{max}} = 1/2$, each edge carries a five-dimensional Hilbert space spanned by the Wigner basis states $\ket{j, m, n}$ with $(j, m, n) \in \{(0, 0, 0)\} \cup\{(1/2, m, n) : m, n = \pm 1/2\}$. The total Hilbert space dimension is $5^4 = 625$. We treat each edge as a $d$-dimensional qudit with $d = 5$ for $j_{\mathrm{max}} = 1/2$; in practice each edge can be encoded into $\lceil \log_2 d \rceil = 3$ qubits by embedding the five-dimensional Hilbert space into a subspace of the eight-dimensional register.

The Kogut--Susskind Hamiltonian for pure SU(2) gauge theory on this plaquette is $H = H_E + H_B$, where $H_E$ is the electric (kinetic) term and $H_B$ is the magnetic (plaquette) term. The electric Hamiltonian is
\begin{equation}\label{eq:H_E}
    H_E = \frac{g^2}{2} \sum_{e \in E}
    \hat{j}_e\bigl(\hat{j}_e + 1\bigr),
\end{equation}
where $\hat{j}_e$ is the Casimir operator on edge $e$, diagonal in the Wigner basis with eigenvalue $j_e(j_e + 1)$. For $j_{\mathrm{max}} = 1/2$ the single-edge electric Hamiltonian is $H_E^{(e)} = (g^2 / 2) \mathrm{diag}(0, 3/4, 3/4, 3/4, 3/4)$, where the entries correspond to $j = 0$ with eigenvalue $0$ and the four $j = 1/2$ states each with eigenvalue $3/4$. The full electric Hamiltonian is the sum $H_E = \sum_{e=0}^{3} H_E^{(e)} \otimes \id_{\bar e}$. Since $H_E$ is diagonal in the Wigner basis, the unitary $\mathrm{e}^{-iH_E\, dt}$ is trivially computed as a diagonal matrix.

The magnetic Hamiltonian involves the trace of the product of link operators around the plaquette,
\begin{equation}\label{eq:H_B}
    H_B = -\frac{1}{g^2}\, \real\!\left[
    \tr_{1/2}\bigl(U_{\mathrm{plaq}}\bigr)\right],
    \qquad
    U_{\mathrm{plaq}} = U_{e_0}\, U_{e_1}\, U_{e_2}\, U_{e_3},
\end{equation}
where $\tr_{1/2}$ denotes the trace in the fundamental ($j = 1/2$) representation. To compute the matrix elements of $H_B$ in the Wigner basis we use the identity for the integral of a product of three Wigner $D$-functions over $\mathrm{SU}(2)$,
\begin{equation}\label{eq:three_D}
    \int_{\mathrm{SU}(2)}
    \overline{D^{(j')}_{m' n'}(g)}\,
    D^{(1/2)}_{a b}(g)\,
    D^{(j)}_{m n}(g)\, dg
    = \frac{C^{j' m'}_{\frac{1}{2} a,\, j m}\;
    C^{j' n'}_{\frac{1}{2} b,\, j n}}{2j' + 1},
\end{equation}
where $C^{J M}_{j_1 m_1,\, j_2 m_2}= \langle j_1, m_1; j_2, m_2 | J, M \rangle$ denotes the Clebsch--Gordan coefficient. This integral arises because each edge contributes one $D$-function from the bra, one from the fundamental-representation link operator, and one from the ket. For each edge, we define the tensor
\begin{equation}\label{eq:edge_tensor}
    T^{(e)}_{a b,\, I' I}
    = \sqrt{d_{j'} d_{j}}\,
    \frac{C^{j' m'}_{\frac{1}{2} a,\, j m}\,
    C^{j' n'}_{\frac{1}{2} b,\, j n}}{d_{j'}},
\end{equation}
where $I = (j, m, n)$ and $I' = (j', m', n')$ are compound indices labeling the single-edge Wigner basis, $a$ and $b$ are fundamental-representation indices contracted around the plaquette, and the factors of $\sqrt{d_j}$ arise from the normalization of the Wigner basis. The selection rules enforce $m' = a + m$, $n' = b + n$, and $|j - 1/2| \le j' \le j + 1/2$. The matrix element of the plaquette trace in the full Wigner basis is then
\begin{equation}\label{eq:plaq_matrix}
    \bra{I_0' I_1' I_2' I_3'}
    \tr_{1/2}\bigl(U_{\mathrm{plaq}}\bigr)
    \ket{I_0 I_1 I_2 I_3}
    = \sum_{a, b_0, b_1, b_2}
    T^{(e_0)}_{a b_0,\, I_0' I_0}\,
    T^{(e_1)}_{b_0 b_1,\, I_1' I_1}\,
    T^{(e_2)}_{b_1 b_2,\, I_2' I_2}\,
    T^{(e_3)}_{b_2 a,\, I_3' I_3},
\end{equation}
where the sum runs over all fundamental-representation indices. The cyclic contraction of the indices $a, b_0, b_1, b_2$ reflects the trace structure of the plaquette. The full magnetic Hamiltonian is obtained by taking the real part and multiplying by $-1/g^2$. We have verified numerically that the resulting $625 \times 625$ matrix $H_B$ is Hermitian, that $[H_B, C^{(v)}] = 0$ for the gauge Casimir $C^{(v)}$ at each vertex, confirming gauge invariance, and that the matrix agrees with an independent Monte Carlo evaluation of the Haar integral to within statistical precision.

\subsection{Trotterization and noise models}
\label{sec:trotter}

The time evolution operator $\mathrm{e}^{-iHt}$ is approximated using a first-order Trotter decomposition. The total evolution time $T$ is divided into $n_{\mathrm{steps}}$ intervals of size $dt = T / n_{\mathrm{steps}}$, and the evolution over each interval is approximated as
\begin{equation}
    \mathrm{e}^{-iH\, dt} \approx
    \mathrm{e}^{-iH_E\, dt}\, \mathrm{e}^{-iH_B\, dt}.
\end{equation}
The Trotter error per step is $O(dt^2 \|[H_E, H_B]\|)$, controlled by choosing $dt$ sufficiently small. We use $T = 3.0$ and $n_{\mathrm{steps}} = 30$, giving $dt = 0.1$. The ideal noiseless evolution is obtained by iterating this Trotter step starting from the strong-coupling vacuum $\ket{\Psi_0} = \ket{0, 0, 0}^{\otimes 4}$. The resulting sequence of states $\{\ket{\Psi_k}\}_{k=0}^{n_{\mathrm{steps}}}$ serves as the reference against which the noisy evolution is compared. We have verified that the Trotterized evolution preserves gauge invariance exactly; $C^{(v)} \ket{\Psi_k} = 0$ at every vertex and every step, as required by the fact that both $H_E$ and $H_B$ commute with the gauge Casimir.

After each Trotter step, noise is applied independently to each edge. The qudit depolarizing channel on a single $d$-dimensional edge is
\begin{equation}\label{eq:depol}
    \mathcal{E}_{\mathrm{depol}}(\rho_e)
    = (1 - p)\, \rho_e + \frac{p}{d}\, \id_d\, \tr(\rho_e),
\end{equation}
where $p \in [0, 1]$ is the error rate per edge per Trotter step. With probability $1 - p$ the state is unchanged and with probability $p$ it is replaced by the maximally mixed state on the edge Hilbert space. This channel is trace-preserving and completely positive for all $p \in [0, 1]$. In the full $625$-dimensional Hilbert space the channel on edge $e$ acts as
\begin{equation}
    \mathcal{E}_{\mathrm{depol}}^{(e)}(\rho)
    = (1 - p)\, \rho + \frac{p}{d}\,
    \tr_e(\rho) \otimes \id_e,
\end{equation}
and the channels on different edges commute. The depolarizing channel mixes all basis states on the affected edge, including states with different $j$ values, and generically creates gauge violations at both endpoints of the edge. We use error rates $p \in \{0.001, 0.005, 0.01\}$, which span the range of effective per-edge error rates achievable on current superconducting quantum hardware; on the latest IBM Heron processors, effective gate error rates are of order $10^{-3}$ to $10^{-2}$~\cite{mckay2023benchmarking}.

The generalized amplitude damping channel models the decay of each edge toward its ground state $\ket{0, 0, 0}$, the $j = 0$ singlet. This is physically motivated by energy relaxation $T_1$ processes in superconducting hardware, generalized to the qudit setting. The channel is defined by Kraus operators on a single $d$-dimensional edge,
\begin{eqnarray}\label{eq:ad_kraus}
    K_0 &=& \ket{0}\!\bra{0}
    + \sqrt{1 - \gamma}\, \sum_{i=1}^{d-1}
    \ket{i}\!\bra{i},
    \\
    K_i &=& \sqrt{\gamma}\, \ket{0}\!\bra{i},
    \qquad i = 1, \ldots, d-1, \nonumber
\end{eqnarray}
where $\ket{0} = \ket{0, 0, 0}$ is the $j = 0$ state, $\ket{i}$ for $i = 1, \ldots, 4$ are the four $j = 1/2$ states, and $\gamma \in [0, 1]$ is the damping rate per edge per Trotter step. The operator $K_0$ preserves the $j = 0$ population and attenuates the $j = 1/2$ amplitudes by $\sqrt{1 - \gamma}$, while $K_i$ transfers population from the $i$-th excited state to $\ket{0, 0, 0}$. The trace-preservation condition $K_0^\dagger K_0 + \sum_{i=1}^{d-1} K_i^\dagger K_i = \id_d$ is straightforward to verify.

\subsection{Convergence of the iterative sweep}
\label{sec:convergence_results}

To study the convergence of iterative gauge cooling, we prepare the strong-coupling vacuum, apply one Trotter step, apply depolarizing noise with $p = 0.005$ on all four edges, and then perform repeated sweeps of gauge cooling. Each sweep consists of syndrome extraction and recovery at vertices $v_0, v_1, v_2, v_3$ in sequence. After each sweep we compute $\mathcal{F}_{\mathrm{GI}}$. The results are shown in table~\ref{tab:convergence}. After a single sweep the gauge-invariant overlap increases modestly, reflecting the fact that each vertex correction partially disrupts its neighbors. Subsequent sweeps reduce the gauge-violating component geometrically with a contraction factor of approximately $0.45$ per sweep. After five sweeps the deficit $1 - \mathcal{F}_{\mathrm{GI}}$ has decreased by roughly an order of magnitude from its initial value.

\begin{table}[h]
\centering
\caption{Gauge-invariant overlap after each sweep of iterative gauge cooling, for a single Trotter step with depolarizing noise $p = 0.005$ on the single-plaquette geometry. The deficit decreases geometrically with a contraction factor of approximately $0.45$ per sweep.}
\label{tab:convergence}
\begin{tabular}{@{}ccc@{}}
\hline
Sweep & $\mathcal{F}_{\mathrm{GI}}$ &
$1 - \mathcal{F}_{\mathrm{GI}}$ \\
\hline
0 (before cooling) & 0.992    & $8.0 \times 10^{-3}$ \\
1                  & 0.993    & $7.0 \times 10^{-3}$ \\
2                  & 0.9957   & $4.3 \times 10^{-3}$ \\
3                  & 0.9977   & $2.3 \times 10^{-3}$ \\
4                  & 0.9988   & $1.2 \times 10^{-3}$ \\
5                  & 0.9994   & $5.5 \times 10^{-4}$ \\
6                  & 0.9997   & $2.5 \times 10^{-4}$ \\
7                  & 0.9999   & $1.2 \times 10^{-4}$ \\
8                  & 0.99995  & $5.2 \times 10^{-5}$ \\
9                  & 0.99998  & $2.3 \times 10^{-5}$ \\
10                 & 0.99999  & $1.0 \times 10^{-5}$ \\
\hline
\end{tabular}
\end{table}

\subsection{Fidelity results}
\label{sec:fidelity_results}

\begin{figure}[t]
\centering
\includegraphics[width=0.95\linewidth]{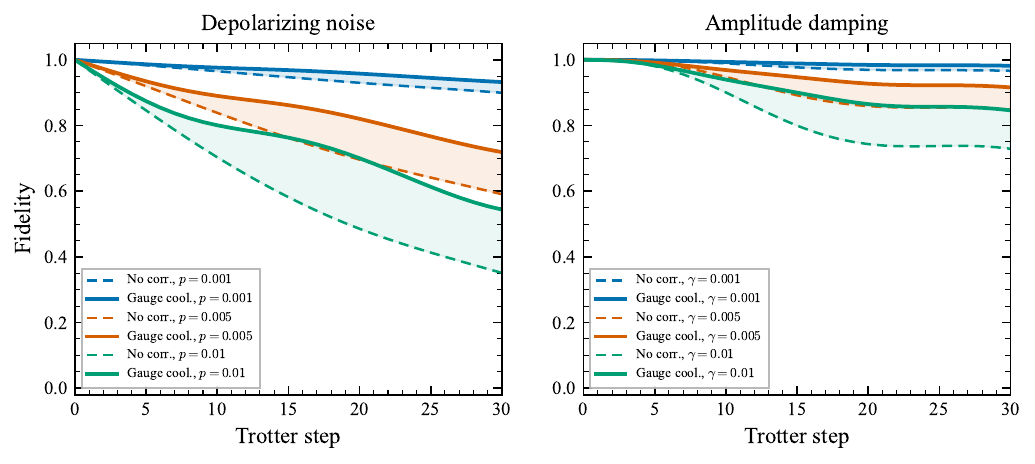}
\caption{Fidelity with the ideal noiseless evolution as a function of Trotter step for a single-plaquette SU(2) lattice gauge theory simulation with $j_{\mathrm{max}} = 1/2$, coupling $g^2 = 1$, total evolution time $T = 3.0$, and step size $dt = 0.1$. Dashed lines, no error correction; solid lines, iterative gauge cooling up to ten sweeps applied after each Trotter step. Left panel, qudit depolarizing noise with error rate $p$ per edge per step; right panel, amplitude damping with damping rate $\gamma$ per edge per step.}
\label{fig:results}
\end{figure}

The fidelity results for both noise channels are shown in figure~\ref{fig:results}. Gauge cooling slows the decay of fidelity across all noise rates tested, with substantial improvement at later time steps. These results establish that active syndrome extraction and recovery can suppress gauge violations and restore approximate gauge invariance in quantum simulations of non-abelian lattice gauge theories on near-term hardware. The extension to larger lattices, where the iterative convergence properties and the interplay with multiplicity-space errors become more significant, is a natural direction for future work.

\section{Conclusion}
\label{sec:conclusion}

We have presented a syndrome extraction and recovery protocol for SU(2) lattice gauge theory on noisy quantum hardware. The protocol uses mid-circuit measurement and a group quantum Fourier transform to extract a triplet syndrome $(J, M, N)$ that resolves the angular momentum and magnetic quantum numbers of any gauge violation at a vertex, and a syndrome-conditional unitary recovery returns the state to the singlet sector. Iterating the procedure as a sweep over vertices defines gauge cooling, which we have shown converges geometrically on the single-plaquette geometry. At coordination-four vertices, where the singlet multiplicity is non-trivial, we have proved that every single-qubit Pauli error is detected by the gauge syndrome and that the Knill--Laflamme conditions for exact recovery fail. The residual physical-subspace errors carry a structured Pauli decomposition whose vanishing $Y$ component suggests compatibility with concatenation by a CSS stabilizer code. Numerical demonstration on a single plaquette shows that gauge cooling restores approximate gauge invariance and improves fidelity under both depolarizing and amplitude damping noise at rates representative of current superconducting hardware.

\bibliographystyle{iopart-num}
\bibliography{refs}

\appendix
\section{Peter--Weyl decomposition and the unitary $t$-design condition}
\label{app:tdesign}

In this appendix we provide the details of the passage from the Peter--Weyl decomposition on individual edges to the Clebsch--Gordan decomposition at each vertex, and we derive the unitary $t$-design condition under which the discrete syndrome extraction is equivalent to the continuum operation.

\subsection*{Vertex regrouping}

By the Peter--Weyl theorem, each link Hilbert space decomposes as in~\eqref{eq:pw_edge}. The full Hilbert space of the lattice, restricting to gauge degrees of freedom, is the tensor product over all edges. Distributing the tensor product over the direct sums, this becomes
\begin{equation}\label{eq:edge_decomp_app}
    \bigoplus_{\{j_e\}_{e \in E}} \bigotimes_{e \in E}
    \bigl(V_{j_e} \otimes V_{j_e}^*\bigr),
\end{equation}
where the direct sum runs over all assignments of a spin label $j_e \in \frac{1}{2}\mathbb{N}_0$ to each edge of the lattice. We regroup the factors in each summand by vertex rather than by edge. An oriented edge $e$ with source vertex $s(e)$ and target vertex $t(e)$ contributes a factor of $V_{j_e}$ that transforms under gauge transformations at $s(e)$ and a factor of $V_{j_e}^*$ that transforms at $t(e)$. This follows from the definition of the gauge action; a gauge transformation $h_v$ at vertex $v$ acts via the left regular representation on outgoing edges, affecting the $m$ index of $V_{j_e}$, and via the right regular representation on incoming edges, affecting the $n$ index of $V_{j_e}^*$.

Collecting all factors associated with a given vertex $v$, we define the vertex Hilbert space
\begin{equation}\label{eq:vertex_hilbert_app}
    \mathcal{H}_v\bigl(\{j_e\}\bigr)
    = \bigotimes_{\substack{e \in E \\ s(e) = v}} V_{j_e}
    \;\otimes\; \bigotimes_{\substack{e \in E \\ t(e) = v}}
    V_{j_e}^*.
\end{equation}
For a fixed assignment of edge spins, the full Hilbert space then factorizes over vertices,
\begin{equation}
    \bigotimes_{e \in E} \bigl(V_{j_e} \otimes V_{j_e}^*\bigr)
    \cong \bigotimes_{v \in V}
    \mathcal{H}_v\bigl(\{j_e\}\bigr).
\end{equation}
This factorization is possible because each edge contributes exactly one tensor factor $V_{j_e}$ to its source vertex and one factor $V_{j_e}^*$ to its target vertex, with every factor thereby assigned to exactly one vertex.

At each vertex $v$ the space $\mathcal{H}_v(\{j_e\})$ is a tensor product of SU(2) irreps, both fundamental and contragradient. Since $V_j^* \cong V_j$ for SU(2), the vertex Hilbert space is equivalent to a tensor product of standard SU(2) representations. The Clebsch--Gordan decomposition then gives~\eqref{eq:vertex_decomp} as in the main text. The gauge transformation $U^{(v)}(h)$ acts as $\pi_{J_v}(h) \otimes \id_{\mu_{J_v}}$ within each summand, so the Gauss law constraint $U^{(v)}(h) = \id$ for all $h$ selects the $J_v = 0$ singlet sector. The physical subspace at vertex $v$ is therefore $V_0 \otimes \mathbb{C}^{\mu_0} \cong \mathbb{C}^{\mu_0}$, with the physical degrees of freedom residing entirely in the multiplicity space.

\subsection*{The $t$-design condition}

A unitary $t$-design on $\mathrm{SU}(2)$ is a finite set $\{g_1, \ldots, g_{n_t}\} \subset \mathrm{SU}(2)$ such that the uniform average over the set reproduces the Haar integral for all polynomials in the matrix entries of $g$ and $\bar g$ of degree at most $t$ in each~\cite{dankert2009,mele2024introduction}. Formally,
\begin{equation}\label{eq:tdesign_def}
    \frac{1}{n_t} \sum_{i=1}^{n_t} f(g_i)
    = \int_{\mathrm{SU}(2)} f(g)\, dg
\end{equation}
for every polynomial $f$ of bidegree at most $(t, t)$ in $g$ and $\bar g$.

The key operator in the syndrome extraction is the discrete $T^{(J)}_{MN}$ defined in~\eqref{eq:T_operator}, which approximates the Haar integral~\eqref{eq:T_haar}. To evaluate when the discrete sum equals the integral exactly, we must determine the polynomial bidegree of the integrand in the matrix entries of $g$ and $\bar g$ separately, since the $t$-design condition imposes independent degree constraints on each.

The factor $\overline{D^{(J)}_{MN}(g)}$ is the complex conjugate of a spin-$J$ representation matrix element. Since $D^{(J)}_{MN}(g) = [\pi_J(g)]_{MN}$ is a polynomial of degree $2J$ in the matrix entries of $g$ and degree $0$ in the entries of $\bar g$, its complex conjugate is a polynomial of degree $0$ in $g$ and degree $2J$ in $\bar g$.

The gauge action $U^{(v)}(g)$ acts on the edges incident to vertex $v$, and the polynomial type of each edge contribution depends on the orientation. On an outgoing edge, where $v$ is the source, the gauge transformation acts via the left regular representation $L_g$, whose matrix elements on the spin-$j_e$ sector are $[\overline{\pi_{j_e}(g)}]_{m' m}$. This follows from $L_g f(h) = f(g^{-1} h)$ together with the unitarity of $\pi_{j_e}$, which gives $[\pi_{j_e}(g^{-1})]_{m m'} = [\pi_{j_e}(g)^\dagger]_{m m'} = [\overline{\pi_{j_e}(g)}]_{m' m}$. Each such factor is a polynomial of degree $0$ in $g$ and degree $2 j_e$ in $\bar g$. On an incoming edge, where $v$ is the target, the gauge transformation acts via the right regular representation $R_g$, whose matrix elements are $[\pi_{j_e}(g)]_{n' n}$, a polynomial of degree $2 j_e$ in $g$ and degree $0$ in $\bar g$. The matrix element of the full integrand between any two basis states is therefore a polynomial with bidegree
\begin{equation}\label{eq:bidegree}
    \deg_g = 2 \sum_{e \in \mathrm{in}(v)} j_e,
    \qquad
    \deg_{\bar g} = 2 J + 2 \sum_{e \in \mathrm{out}(v)} j_e.
\end{equation}

\begin{lemma}[$t$-design condition]
\label{lem:tdesign}
The discrete $T^{(J)}_{MN}$ in~\eqref{eq:T_operator} equals the continuum operator in~\eqref{eq:T_haar} exactly provided
\begin{equation}\label{eq:t_requirement}
    t \;\ge\; \max\!\left(
    2 \sum_{e \in \mathrm{in}(v)} j_e,\;
    2 J + 2 \sum_{e \in \mathrm{out}(v)} j_e
    \right),
\end{equation}
where $\mathrm{in}(v)$ and $\mathrm{out}(v)$ denote the sets of incoming and outgoing edges at $v$.
\end{lemma}

\begin{proof}
The condition follows immediately from the bidegree calculation~\eqref{eq:bidegree} and the $t$-design definition~\eqref{eq:tdesign_def}; the discrete sum equals the Haar integral whenever the integrand is a polynomial of bidegree at most $(t, t)$ in $(g, \bar g)$.
\end{proof}

Since the truncation restricts $J \le j_{\mathrm{cut}}$, each edge spin satisfies $j_e \le j_{\mathrm{max}}$, and the maximum total angular momentum at a vertex with $k$ incident edges is $J_{\mathrm{max}} = k \cdot j_{\mathrm{max}}$, the bound~\eqref{eq:t_requirement} becomes
\begin{equation}\label{eq:t_sufficient_general}
    t \;\ge\; \max\!\left(
    2 k_{\mathrm{in}}\, j_{\mathrm{max}},\;
    2 j_{\mathrm{cut}} + 2 k_{\mathrm{out}}\, j_{\mathrm{max}}
    \right).
\end{equation}
Setting $j_{\mathrm{cut}} = k \cdot j_{\mathrm{max}}$, which is the minimal value for which the syndrome extraction resolves all gauge-violation sectors (appendix~\ref{app:qft}), and noting that the second argument then dominates since $k = k_{\mathrm{in}} + k_{\mathrm{out}}$, a sufficient condition is
\begin{equation}\label{eq:t_sufficient}
    t \;\ge\; 2 k\, j_{\mathrm{max}}
    + 2 k_{\mathrm{out}}\, j_{\mathrm{max}}.
\end{equation}
For the single-plaquette geometry with $k = 2$, $k_{\mathrm{out}} = 1$, $j_{\mathrm{max}} = 1/2$, this gives $t \ge 3$. For a square lattice vertex with $k = 4$, $k_{\mathrm{out}} = 2$, $j_{\mathrm{max}} = 1/2$, the requirement is $t \ge 6$.

When the $t$-design condition is satisfied, the discrete and continuum protocols produce identical measurement statistics and identical post-measurement states on the data register, so the syndrome extraction is faithful. The size $n_t$ of the $t$-design determines the dimension of the ancilla register; a $t$-design on $\mathrm{SU}(2)$ requires at least $\sum_{j=0,\frac12,1,\ldots}^{t} (2j+1)^2 = O(t^3)$ elements by dimension counting of the Peter--Weyl basis up to spin $t$, so the ancilla overhead grows polynomially with the coordination number and $j_{\mathrm{max}}$.

\section{Truncation of the group quantum Fourier transform}
\label{app:qft}

The group quantum Fourier transform $\qft_{\mathrm{SU}(2)}$ as defined in~\eqref{eq:qft} acts on an infinite-dimensional Hilbert space, since the Peter--Weyl decomposition of $L^2(\mathrm{SU}(2))$ involves all spins $j \in \frac{1}{2} \mathbb{N}_0$. In this appendix we justify the truncation to $j \le j_{\mathrm{cut}}$ and characterize the resulting map.

The truncated QFT is defined by~\eqref{eq:qft_truncated}. Its output space is spanned by $\{\ket{j, m, n} : j \le j_{\mathrm{cut}},\, -j \le m, n \le j\}$ and has dimension
\begin{equation}\label{eq:dout}
    d_{\mathrm{out}} = \sum_{j=0,\frac12,1,\ldots}^{j_{\mathrm{cut}}} (2j+1)^2.
\end{equation}
For $j_{\mathrm{cut}} = 1/2$ this gives $d_{\mathrm{out}} = 1 + 4 = 5$. The input space is the $n_t$-dimensional ancilla register spanned by $\{\ket{g_1}, \ldots, \ket{g_{n_t}}\}$. For the truncated map to be well defined we need $n_t \ge d_{\mathrm{out}}$. When $n_t > d_{\mathrm{out}}$, the map has a non-trivial kernel; when $n_t = d_{\mathrm{out}}$, it is a square matrix and can be unitary if the $t$-design elements are chosen appropriately.

\begin{lemma}[Isometry property]
\label{lem:isometry}
Under the $t$-design condition of lemma~\ref{lem:tdesign}, the truncated QFT acts as an isometry from the ancilla register into the output subspace on all states within the truncated edge Hilbert space.
\end{lemma}

\begin{proof}
Consider the inner product of two output states,
\begin{eqnarray}\label{eq:inner_product}
    \bra{g_i} \bigl(\qft^{\le j_{\mathrm{cut}}}\bigr)^\dagger
    \qft^{\le j_{\mathrm{cut}}} \ket{g_k}
    &=& \sum_{j=0}^{j_{\mathrm{cut}}} \frac{2j + 1}{n_t}
    \sum_{m, n} [\pi_j(g_i)]_{m n}\,
    [\overline{\pi_j(g_k)}]_{m n} \nonumber \\
    &=& \sum_{j=0}^{j_{\mathrm{cut}}} \frac{2j + 1}{n_t}\,
    \chi_j\bigl(g_i g_k^{-1}\bigr),
\end{eqnarray}
where $\chi_j(g) = \tr[\pi_j(g)]$ is the character of the spin-$j$ representation. In the continuum ($j_{\mathrm{cut}} \to \infty$), the Peter--Weyl completeness relation gives $\sum_j (2j+1)\, \chi_j(g_i g_k^{-1}) = \delta(g_i, g_k)$, recovering the orthonormality of the distributional states $\langle g | g' \rangle = \delta(g, g')$. On the discrete ancilla register the $t$-design condition replaces the Haar integral with a finite sum, introducing a factor of $n_t$ that is absorbed by the $1/\sqrt{n_t}$ prefactor in~\eqref{eq:qft_truncated}.
\end{proof}

The syndrome extraction does not require the full delta function. The operator $T^{(J)}_{MN}$ acts within the $J$-isotypic component of $\mathcal{H}_v$, and the vertex Hilbert space is itself truncated to $j_e \le j_{\mathrm{max}}$ on each edge. The Clebsch--Gordan decomposition at a vertex with $k$ edges carrying spins at most $j_{\mathrm{max}}$ produces total angular momenta $J_v \le k\, j_{\mathrm{max}}$. Therefore the syndrome extraction is exact provided
\begin{equation}\label{eq:jcut_condition}
    j_{\mathrm{cut}} \ge k\, j_{\mathrm{max}},
\end{equation}
which ensures that all gauge-violation sectors that can arise from the truncated edge Hilbert spaces are resolved by the measurement. For the single-plaquette geometry with $k = 2$ and $j_{\mathrm{max}} = 1/2$ this gives $j_{\mathrm{cut}} \ge 1$, so the truncated QFT with $j_{\mathrm{cut}} = 1$ is sufficient.

When $n_t > d_{\mathrm{out}}$ the truncated QFT is a rectangular matrix of size $d_{\mathrm{out}} \times n_t$. To implement it as a quantum circuit we embed it into a unitary on the full $n_t$-dimensional ancilla register. Let $W$ denote the $d_{\mathrm{out}} \times n_t$ matrix of the truncated QFT. Since $W$ is an isometry by lemma~\ref{lem:isometry}, with $W W^\dagger = \id_{d_{\mathrm{out}}}$ up to truncation corrections, we extend it to a unitary $\tilde W$ of size $n_t \times n_t$ by appending $n_t - d_{\mathrm{out}}$ orthonormal rows that span the orthogonal complement of the range of $W$. The choice of these additional rows does not affect the syndrome extraction; measurement outcomes in the $\ket{j, m, n}$ basis with $j \le j_{\mathrm{cut}}$ are determined entirely by $W$, and outcomes corresponding to the orthogonal complement indicate that the ancilla has landed outside the truncated Wigner basis. Such outcomes can be treated as a no-information result, after which no recovery is applied. For the parameter regime used in this work, with $j_{\mathrm{cut}} \ge k\, j_{\mathrm{max}}$, such outcomes have exactly zero probability.

\section{Recovery operator and basis assignment on the single plaquette}
\label{app:recovery}

In this appendix we describe the construction of the recovery operator on the single-plaquette geometry, including the Clebsch--Gordan decomposition at coordination-two vertices and the spectator-preserving basis assignment that ensures consistency between syndrome extraction and recovery.

\subsection*{Clebsch--Gordan decomposition at coordination-two vertices}

At each vertex of the single plaquette the gauge transformation $U^{(v)}(h)$ acts on two indices, namely the $m$ index of the outgoing edge via the left regular representation and the $n$ index of the incoming edge via the right regular representation. The vertex Hilbert space for fixed edge spins $j_{\mathrm{out}}$ and $j_{\mathrm{in}}$ is therefore $V_{j_{\mathrm{out}}} \otimes V_{j_{\mathrm{in}}}^*$. Since $V_j^* \cong V_j$ for SU(2) this is equivalent to coupling two angular momenta. For $j_{\mathrm{max}} = 1/2$ the possible combinations and their Clebsch--Gordan decompositions are listed in table~\ref{tab:cg_coord2}. The singlet $J = 0$ sector appears in the first and fourth rows, each with multiplicity $\mu_0 = 1$. These live in different edge-spin blocks that are not mixed by gauge transformations.

\begin{table}[h]
\centering
\caption{Clebsch--Gordan decomposition at a coordination-two vertex for $j_{\mathrm{max}} = 1/2$.}
\label{tab:cg_coord2}
\begin{tabular}{@{}ccc@{}}
\hline
$j_{\mathrm{out}}$ & $j_{\mathrm{in}}$ & Decomposition \\
\hline
$0$   & $0$   & $J = 0$ \quad ($\mu_0 = 1$) \\
$0$   & $1/2$ & $J = 1/2$ \quad ($\mu_{1/2} = 1$) \\
$1/2$ & $0$   & $J = 1/2$ \quad ($\mu_{1/2} = 1$) \\
$1/2$ & $1/2$ & $J = 0 \oplus J = 1$ \quad
($\mu_0 = 1,\, \mu_1 = 1$) \\
\hline
\end{tabular}
\end{table}

The vertex Hilbert space $\mathcal{H}_v$ collects only the indices that transform under the gauge action at $v$, which are $V_{j_{\mathrm{out}}}$ from the outgoing edge and $V_{j_{\mathrm{in}}}^*$ from the incoming edge. Summing over all edge-spin assignments, the total dimension is $1 + 2 + 2 + 4 = 9$. The $J$ sectors across this nine-dimensional space are $J = 0$ with total multiplicity $2$ from the $(0, 0)$ and $(1/2, 1/2)$ blocks, $J = 1/2$ with total multiplicity $2$ from the $(0, 1/2)$ and $(1/2, 0)$ blocks each carrying $2J + 1 = 2$ states, and $J = 1$ with multiplicity $1$ from the $(1/2, 1/2)$ block carrying $2J + 1 = 3$ states. This accounts for all $2 + 4 + 3 = 9$ states.

In practice the CG basis is constructed numerically by building the gauge generators $G_a^{(v)}$ ($a = x, y, z$) at each vertex, forming the Casimir operator $C^{(v)} = \sum_a (G_a^{(v)})^2$, and diagonalizing $C^{(v)}$ and $G_z^{(v)}$ simultaneously. The eigenvalues of $C^{(v)}$ are $J(J + 1)$, giving the $J$ quantum number, and the eigenvalues of $G_z^{(v)}$ within each $J$ sector give $M$. The gauge generators act on the five-dimensional edge space as follows. On the outgoing edge the generator of left multiplication $L_a$ acts on the $m$ index within each $(j, n)$ block,
\begin{equation}
    L_a \ket{j, m, n} = \sum_{m'} [\mathcal{J}_a^{(j)}]_{m' m}
    \ket{j, m', n},
\end{equation}
where $\mathcal{J}_a^{(j)}$ is the spin-$j$ representation of the $\mathfrak{su}(2)$ generator $J_a$. Note that because $L_h f(g) = f(h^{-1} g)$, the generators of $L$ carry a sign relative to the representation matrices, $\mathcal{J}_a^{(L)} = -J_a^*$. On the incoming edge the generator of right multiplication $R_a$ acts on the $n$ index,
\begin{equation}
    R_a \ket{j, m, n} = \sum_{n'} [J_a^{(j)}]_{n' n}
    \ket{j, m, n'}.
\end{equation}
The full gauge generator at vertex $v$ is $G_a^{(v)} = L_a^{(e_{\mathrm{out}})} \otimes \id_{e_{\mathrm{in}}} + \id_{e_{\mathrm{out}}} \otimes R_a^{(e_{\mathrm{in}})}$, embedded into the full $625$-dimensional space by tensoring with the identity on the remaining two edges.

\subsection*{Recovery operator construction}

The gauge cooling recovery operator $R_{J, M}$ maps the subspace $\mathcal{W}_M^J$ back to $\mathcal{W}_0^0$ at each vertex. For the single-plaquette geometry, where $\mu_0 = 1$ within each edge-spin block, the recovery is uniquely determined up to a phase; it maps the single basis vector in $\mathcal{W}_M^J$, for a given edge-spin assignment, to the single singlet state in $\mathcal{W}_0^0$.

The construction proceeds as follows. Having obtained the CG basis by diagonalizing the Casimir and $G_z$, we organize the basis vectors by $J$ sector. Within each $J$ sector, the $M$ values range from $-J$ to $J$, and we need the basis vectors at different $M$ values to be consistently related by the gauge raising and lowering operators. This ensures that the recovery operator preserves the multiplicity structure. Starting from the lowest-weight states $M = -J$ within each $J$ sector, we apply the raising operator $G_+ = G_x + i G_y$ to generate all higher-$M$ states,
\begin{equation}\label{eq:raising}
    \ket{J, M+1, \alpha}
    = \frac{G_+ \ket{J, M, \alpha}}
    {\sqrt{J(J+1) - M(M+1)}},
\end{equation}
where $\alpha$ labels the multiplicity. This produces a consistently labeled basis $\{\ket{J, M, \alpha}\}$ in which states at different $M$ with the same $\alpha$ are related by the standard angular momentum algebra. The Kraus operators for the gauge cooling channel then take the form~\eqref{eq:kraus_gc} with the same multiplicity indexing.

In practice, the recovery can be further refined by matching each non-singlet state $\ket{J, N, \alpha}$ to the singlet state $\ket{0, 0, \alpha'}$ that shares the same spectator quantum numbers, namely the edge-spin labels and indices not acted on by the gauge transformation at $v$. This spectator-preserving assignment minimizes changes to the edge states, reducing disruption to neighboring vertices that share edges and improving the convergence of the iterative sweep.

\section{Group algebra basis and Wigner basis}
\label{app:bases}

We close with a remark on the relationship between the two bases that appear in the syndrome extraction protocol; the group algebra basis $\ket{g}$ used in steps 1 and 2 and the Wigner basis $\ket{j, m, n}$ used in the measurement of step 4. A state $\psi \in L^2(\mathrm{SU}(2))$ is a square-integrable function $\psi : \mathrm{SU}(2) \to \mathbb{C}$. The Peter--Weyl theorem provides a complete orthonormal basis for this space in terms of the normalized Wigner matrix elements $\sqrt{2j+1}\, [\pi_j(g)]_{m n}$, so that any state can be expanded as
\begin{equation}\label{eq:pw_expansion}
    \psi(g) = \sum_j \sqrt{2j+1}
    \sum_{m, n = -j}^{j} \hat\psi^{(j)}_{m n}\,
    D^{(j)}_{m n}(g),
\end{equation}
where the Peter--Weyl coefficients are given by
\begin{equation}\label{eq:pw_coefficients}
    \hat\psi^{(j)}_{m n}
    = \sqrt{2j+1} \int_{\mathrm{SU}(2)} \psi(g)\,
    \overline{D^{(j)}_{m n}(g)}\, dg.
\end{equation}
The group quantum Fourier transform $\qft_{\mathrm{SU}(2)}$ is the unitary change of basis from the group element picture to the representation picture; it takes a state described by its values $\psi(g)$ and re-expresses it in terms of its Peter--Weyl coefficients $\hat\psi^{(j)}_{m n}$.

The group element $\ket{g}$ is not a state in $L^2(\mathrm{SU}(2))$ but rather a distributional object satisfying $\langle g | g' \rangle = \delta(g, g')$ with respect to the Haar measure, analogous to the position eigenstate $\ket{x}$ in ordinary quantum mechanics. The value of a state at a group element is recovered via $\psi(g) = \langle g | \psi \rangle$, exactly as $\psi(x) = \langle x | \psi \rangle$ in the position representation. The analogy with ordinary quantum mechanics is precise; the position basis $\ket{x}$ corresponds to the group element basis $\ket{g}$, the momentum basis $\ket{p}$ corresponds to the Wigner basis $\ket{j, m, n}$, the Fourier transform corresponds to $\qft_{\mathrm{SU}(2)}$, and the zero-momentum state $\ket{p = 0} = \int \ket{x}\, dx$ corresponds to the gauge-invariant state $\ket{0, 0, 0} = \int \ket{g}\, dg$. A uniform superposition over all group elements corresponds to a function that is constant on $\mathrm{SU}(2)$. Its Peter--Weyl expansion has support only on $j = 0$, since the constant function is precisely the trivial representation. In the Wigner basis this state is $\ket{j = 0, m = 0, n = 0}$.

In the syndrome extraction circuit this duality is exploited as follows. The preparation $F_G$ creates a uniform superposition over group elements, or in the truncated setting over a $t$-design, expressing the $J_v = 0$ state in the group algebra basis. This basis is chosen because the controlled gauge action $\sum_i \ket{g_i}\!\bra{g_i} \otimes U^{(v)}(g_i)$ is diagonal in the group element basis, making the controlled operation straightforward to implement. The $\qft_{\mathrm{SU}(2)}$ then transforms the ancilla to the Wigner basis, where the measurement in $\ket{j, m, n}$ extracts the syndrome $(J, M, N)$. The circuit thus never works in a single basis throughout; it uses the group algebra basis where the controlled-$U$ is natural and the Wigner basis where the measurement is natural, with the group QFT serving as the bridge between the two.

\end{document}